\documentclass[a4paper,10pt,preprint,3p]{elsarticle}

\usepackage{lineno,hyperref}
\modulolinenumbers[5]
\usepackage[ruled,vlined]{algorithm2e}
\usepackage[american]{babel}	
\usepackage[utf8]{inputenc}		
\usepackage[T1]{fontenc}		
\usepackage{lmodern}
\usepackage[english=american]{csquotes}	
\usepackage{graphicx}			
\usepackage{epstopdf}
\usepackage{marvosym}			
\usepackage{amsmath}			
\usepackage{amssymb}
\usepackage{amsthm}
\usepackage{mathtools}
\usepackage{color}

\usepackage[normalem]{ulem}
\usepackage{nicefrac}
\usepackage{soul}
\usepackage{amsfonts}    	
\usepackage{xcolor}
\usepackage{csquotes}
\usepackage{float}
\usepackage{url}
\usepackage[format=plain,labelfont=bf]{caption}
\usepackage{capt-of}
\usepackage{subcaption}
\theoremstyle{definition}

\theoremstyle{plain}

\newtheorem{theorem}{Theorem}

\newtheorem{corollary}{Corollary}

\newcommand{\blue}{\textcolor{black}}

\begin{document}

\begin{frontmatter}
\title{Non-binary universal tree-based networks}

\author[label1]{Mareike Fischer\corref{cor1}}
\ead{email@mareikefischer.de}
\cortext[cor1]{Corresponding author}
\author[label1]{Michelle Galla}
\author[label1]{Kristina Wicke}

\address[label1]{Institute of Mathematics and Computer Science, University of Greifswald, Greifswald, Germany}

\begin{abstract} 
A tree-based network $N$ on $X$ is called universal if every phylogenetic tree on $X$ is a base tree for $N$. Recently, binary universal tree-based networks have attracted great attention in the literature and their existence has been analyzed 
\blue{in various studies}. In this note, we extend the analysis to non-binary networks and show that there exist both a rooted and an unrooted non-binary universal tree-based network with $n$ leaves for all positive integers $n$.
\end{abstract}

\end{frontmatter}

\section{Introduction and Preliminaries}
Phylogenetic networks are a generalization of phylogenetic trees allowing for the representation of reticulate evolutionary events such as hybridization and horizontal gene transfer. Even though the existence of reticulate evolutionary events is widely accepted, it has been argued that evolution is still fundamentally tree-like with some occasional events of e.g. horizontal gene transfer. This has led to the introduction of so-called tree-based networks, i.e. networks that can be obtained from a tree by adding additional edges (cf. \citet{Francis2015}). In the following, we consider the existence of one particular class of tree-based networks, namely universal tree-based networks. However, we start by introducing some definitions and notations.

Let $X=\{1, \ldots, n\}$ be a non-empty finite set (e.g. of taxa or species). 

An \emph{unrooted phylogenetic network} $N$ on $X$ is a connected, simple graph $G=(V,E)$ with $X \subseteq V$ and no degree-2 vertices, where the set of degree-1 vertices (referred to as the \emph{leaves} of the network) is bijectively labeled and thus identified with $X$. Such an unrooted network is called \emph{unrooted binary} if every non-leaf vertex has degree 3. We call an unrooted network \emph{unrooted non-binary} if it is not necessarily binary.

A \emph{rooted phylogenetic network} $N$ on $X$ is a rooted directed acyclic graph $G=(V,E)$ with no parallel edges satisfying the following properties:
\begin{enumerate}[(i)]
\item the root has indegree 0 and outdegree 2 or more;
\item all vertices with outdegree 0 have indegree 1, and the set of vertices with outdegree 0 is identified with $X$ (and the vertices in $X$ are again called \emph{leaves});
\item all other vertices either have indegree 1 and outdegree 2 or more (in which case they are called \emph{tree vertices}) or indegree 2 or more and outdegree 1 (in which case they are called \emph{reticulations} or \emph{reticulation vertices}).
\end{enumerate}
If the root has outdegree 2, all reticulations have indegree 2 and all tree-vertices have outdegree 2, the network is called \emph{rooted binary}. We call a rooted network \emph{rooted non-binary} if it is not necessarily binary.
For technical reasons, if $\vert X \vert = 1$, we allow an (un)rooted network to consist of a single leaf (which in case of a rooted network is then at the same time considered to be the root). 
Moreover, note than an (un)rooted \emph{phylogenetic tree} on $X$ is an (un)rooted phylogenetic network whose underlying graph structure is a tree. Furthermore, we refer to 
\blue{non-phylogenetic, i.e. non-leaf-labeled, trees as }\emph{tree shapes}.

An (un)rooted phylogenetic network $N$ on $X$ is called \emph{tree-based} if there is a spanning tree $T = (V,E')$ for $N$ (with $E' \subseteq E$) whose leaf set is equal to $X$. $T$ is then called a \emph{support tree} (for $N$). Moreover, the tree $T'$ that can be obtained from $T$ by suppressing potential degree-2 vertices is called a \emph{base tree} (for $N$). Note that the existence of a support tree $T$ for $N$ implies the existence of a base tree $T'$ for $N$. Moreover, note that from a graph-theoretical point of view, the support tree $T$ is a \emph{subdivision} of the base tree $T'$ (as $T$ can be obtained from $T'$ by subdividing edges of $T$ with degree-2 vertices).

Now, an (un)rooted tree-based network $N$ on $X$ is called \emph{universal} if every 
(un)rooted phylogenetic tree on $X$ is a base tree for $N$.

Universal tree-based networks have recently gained considerable interest in the literature. First, \citet{Francis2015}, who introduced the class of tree-based networks, showed that there is a rooted binary universal tree-based network on $X$ for $n=3$ and asked whether such a network exists in general. This was answered affirmatively by \citet{Hayamizu2016} and, independently, by \citet{Zhang2016}, who gave explicit constructions for such networks for all positive integers $n$. While the construction in \citet{Hayamizu2016} contains $\Theta(n!)$ reticulations, the construction in \citet{Zhang2016} contains $\Theta(n^2)$ reticulations. 
More recently, \citet{Bordewich2017} showed that a rooted binary universal tree-based network with $n$ leaves has $\Omega(n \log (n))$ reticulations and gave a construction of a rooted binary universal tree-based network with $O(n \log (n))$ reticulations.
Moreover, \citet{Francis2018} have recently shown that the existence of a rooted binary universal tree-based network implies the existence of an unrooted binary universal tree-based network. Note, however, that so far all considerations of universal tree-based networks in the literature have been concerned with binary networks. In this note, we extend these findings to non-binary networks and constructively show that there exist both a rooted and an unrooted non-binary universal tree-based network with $n$ leaves for all positive integers $n$.

\section{Results}
\subsection{Rooted universal tree-based networks} \label{rooted_universal}

\begin{theorem} \label{thm_rooted_universal}
For all positive integers $n$, there exists a rooted non-binary universal tree-based network with $n$ leaves.
\end{theorem}

In the proof of Theorem \ref{thm_rooted_universal} we will present a construction of a rooted tree-based network $N$ on $X$ for each $n$. Following the constructions in \citet{Hayamizu2016}, \citet{Zhang2016} and \citet{Bordewich2017}, this construction consists of two parts: the upper part, which contains the root, is a network on $n$ leaves that has every non-binary tree shape on $n$ leaves as a base tree; the lower part, which contains the leaves, reorders the leaves of these tree shapes, in order to enable any permutation of leaves and thus, to enable every binary or non-binary phylogenetic tree on $X$ to be a base tree for this network.
For the latter, one can for example use a so-called \emph{Bene\v{s} network} (cf. \citet{Benes1964a, Benes1964}) as in \citet{Bordewich2017}. Alternatively, one can simply use a complete bipartite graph $K_{n,n}$ for this purpose (cf. Figure \ref{Fig_Embedding}(a) and (b)).

Thus, in the following we will only show that the upper part of our construction has every tree shape as a base tree. Analogously to \citet{Bordewich2017} it then follows that the combination of the upper part with a Bene\v{s} network or a complete bipartite graph yields a universal tree-based network.

\begin{proof}[Proof (Theorem \ref{thm_rooted_universal})]
For all positive integers $n$, we now give a construction of a rooted phylogenetic network $U_n$ on $n$ leaves that has every tree shape on $n$ leaves as base tree. 
We begin by describing the construction of $U_n$. 
First of all, for $n=1$, $U_n$ consists of a single vertex. Now, let $n \geq 2$. 
Then, in order to construct $U_n$, we start with a rooted star tree $T^\ast$ with root $\rho$ on $n$ leaves where a rooted star tree is a rooted tree shape such that all leaves are adjacent to the root and:
	\begin{enumerate}
	\item Add attachment points to the edges of $T^\ast$ as follows:
		\begin{enumerate}
		\item [--] For leaf $1$ and $n$, add $n-2$ attachment points
		on the edges $(\rho,1)$ and $(\rho, n)$, respectively, and label them $t_1^1, t_1^2, \ldots, t_1^{n-2}$  and $t_n^1, t_n^2, \ldots, t_n^{n-2}$, respectively (starting the labeling at the attachment point closest to the root; note that these vertices will be tree vertices in the final network);
		\item [--] For each leaf $l=2,3, \ldots, n-2$, add $2n-4$ attachment points on the edge $(\rho, l)$ and label them $r_l^1, t_l^1, r_l^2, t_l^2, \ldots,  r_l^{n-2}, t_l^{n-2}$ (again, starting the labeling at the attachment point closest to the root; note that all vertices labeled with \enquote{$r$} will be reticulation vertices in the final network and all vertices labeled with \enquote{$t$} will be tree vertices);
		\item [--]  For leaf $n-1$, add $2n-5$ attachment points on edge $(\rho, n-1)$ and label them $r_{n-1}^1, t_{n-1}^1, r_{n-1}^2, t_{n-1}^2,$ $\ldots, t_{n-1}^{n-3}, r_{n-1}^{n-2}$ (again, starting the labeling at the attachment point closest to the root; note that there is no attachment point $t_{n-1}^{n-2}$, however, all vertices labeled with \enquote{$r$} will be reticulation vertices in the final network and all vertices labeled with \enquote{$t$} will be tree vertices).
		\end{enumerate}
	\item Add the following edges between attachment points: \label{edges} 
		\begin{enumerate}
		\item [--]  $(t_i^k, r_{i+1}^k)$ for $i=2, \ldots, n-2$ and $k=1, \ldots, n-2$ (horizontal edges between tree and reticulation vertices);
		\item [--]  $(t_i^k,r_j^{k+1})$ for $i=2, \ldots, n-2$, $j=i+1, \ldots, n-1$ and $k=1, \ldots, n-3$ (diagonal edges between tree vertices and reticulation vertices from left to right);
		\item [--]  $(t_i^k,r_j^{k+1})$ for $i=3, \ldots, n-1$, $j=2, \ldots, i-1$ and $k=1, \ldots, n-3$ (diagonal edges between tree vertices and reticulation vertices from right to left);
		\item [--] $(t_i^k, r_j^k)$ for $i \in \{1,n\}$,  $j=2, \ldots, n-1$ and $k=1, \ldots, n-2$ (diagonal edges between tree vertices on the paths from $\rho$ to leaves $1$ and $n$, respectively, and reticulation vertices on the paths from $\rho$ to leaves $2, \ldots, n-1$).
		\end{enumerate}
	\end{enumerate}

Figure~\ref{Fig_Embedding}(e) shows the resulting construction for $n=3$ and Figure~\ref{Fig_Embedding}(i) shows the construction for $n=5$. 

We now use induction on $n$ to show that -- ignoring the leaf labels -- every tree shape on $n$ leaves is base tree of $U_n$. Since there is exactly one tree shape for $n=1$ (consisting of a single vertex) and $n=2$, the base case holds for all $n \leq 2$. 
Now, suppose that the claim holds for up to $n-1$ leaves and consider the network $U_n$ on $n$ leaves (with $n \geq 3$).

Note that as the basic structure of $U_n$ is a star tree on $n$ leaves, the star tree on $n$ trivially is a base tree of $U_n$. Therefore, we will now show that any other tree shape on $n$ leaves is also a base tree of $U_n$.

Thus, let $T_n$ be an arbitrary tree shape (that is not a star tree) with $n$ leaves and root $\rho$.
We will now show that $T_n$ is a base tree of $U_n$ by constructing an explicit embedding of $T_n$ into $U_n$. i.e., we construct a subdivision of $T_n$ covering all vertices of $U_n$ such that the leaf sets of $T_n$ and $U_n$ coincide.
As $n \geq 3$, we know that $T_n$ contains at least one cherry $[u,v]$, i.e. a pair of leaves $u$ and $v$ who share a common parent (cf. Proposition 1.2.5 in \citet{Semple2003}), say $w$. As, by assumption, $T_n$ is not a star tree, we may assume that $w \neq \rho$. Suppose that $w$ has $k \geq 2$, children in total (including $u$ and $v$). Moreover, without loss of generality we may assume that the children of $w$ are labeled $1, \ldots, k$ when enumerating all leaves and are positioned at the outermost left of the tree when drawing it in the plane.
We now delete all children of $w$ (which implies that $w$ is now a leaf) and retrieve a tree shape $T_{n-k+1}$ with $n-k+1$ leaves. 
As $n-k+1 < n$, by induction $T_{n-k+1}$ is a base tree of $U_{n-k+1}$. Let $V(T_{n-k+1})$ and $E(T_{n-k+1})$ denote the vertex set, respectively edge set, of the underlying embedding of $T_{n-k+1}$ into $U_{n-k+1}$.
In the following, we will first show that $T_{n-k+1}$ can also be embedded in $U_n$; we will then re-introduce the deleted children of $w$ and show that this yields a base tree of $U_n$. 

Note that by construction $T_{n-k+1}$ contains leaves labeled with $w$, $k+1$, $k+2, \ldots, n$. Before we embed $T_{n-k+1}$ into $U_n$, we relabel some vertices of $U_{n-k+1}$. To be precise, we rename vertices (if they exist) as follows (as an example see Figure \ref{Fig_Embedding}(e)):
\begin{align*}
w &\hookrightarrow 1 \\
t_{n-k+1}^{l} &\hookrightarrow t_{n}^{l} \, \text{for } \, l=1, \ldots, {n-k-1} \\
r_{j}^{l} &\hookrightarrow r_{j+k-1}^{l} \text{ and } t_{j}^{l} \hookrightarrow t_{j+k-1}^{l} 
\text{for } j=2, \ldots, n-k \text{ and } l=1, \ldots, {n-k-1}.
\end{align*}

We now sequentially extend the network $U_{n-k+1}$ to $U_n$ by introducing additional vertices and edges.

First of all, we add attachment points on existing edges (cf. Figure \ref{Fig_Embedding}(f)), where for technical reasons $t_1^0 = t_n^0 = \rho$.
	\begin{enumerate}
	\item [--]  edge $(t_1^{n-k-1}, 1)$: $k-1$ attachment points $t_1^{n-k}, t_1^{n-k+1}, \ldots, t_1^{n-2}$;
	\item [--] edge $(t_n^{n-k-1},n)$: $k-1$ attachment points $t_n^{n-k}, t_n^{n-k+1}, \ldots, t_n^{n-2}$;
	\item [--] edge $(r_{n-1}^{n-k-1}, n-1)$ (if it exists): $2(k-1)=2k-2$ attachment points $t_{n-1}^{n-k-1}, r_{n-1}^{n-k}, t_{n-1}^{n-k}, \ldots, r_{n-1}^{n-2}$;
	\item [--] edge $(t_{j+k-1}^{n-k-1}, j+k-1)$ for $j=2, \ldots, n-k-1$ (if it exists): $2(k-1)=2k-2$ attachment points $r_{j+k-1}^{n-k}, t_{j+k-1}^{n-k}, r_{j+k-1}^{n-k+1}, t_{j+k-1}^{n-k+1}, \ldots r_{j+k-1}^{n-2}, t_{j+k-1}^{n-2} $
\end{enumerate}

We add all newly introduced edges to $E(T_{n-k+1})$, i.e. we extend the embedding of  $T_{n-k+1}$ to cover all newly introduced attachment points (as an example see Figure \ref{Fig_Embedding}(f)). 

We then add $k-1$ edges connecting the root to leaves $2, \ldots, k$ and on each of these edges add $2n-4$ attachment points called $r_i^{1}, t_i^1, \ldots, r_i^{n-2}, t_i^{n-2}$ for $i=2, \ldots, k$ (as an example see Figure \ref{Fig_Embedding}(g)). If $k=n-1$, we only add $2n-5$ attachment points on edge $(\rho, {n-1})$. In particular, we do not add the attachment point $t_{n-k}^{n-2}$.

In order to complete the construction of $U_n$, we add all required edges between newly introduced vertices, i.e. we complete the construction of $U_n$ according to the construction principle presented at the beginning of the proof (see page \pageref{edges}; as an example see Figure \ref{Fig_Embedding}(h)).

We now re-introduce the children of $w$ to the embedding of $T_{n-k+1}$ in order to obtain an embedding of $T_n$, i.e. we re-introduce the leaves $2, \ldots, k$ (note that we do not re-introduce leaf $1$, as this was already re-introduced in a previous step). We do this in the following way (cf. Algorithm \ref{alg}): 

\begin{algorithm}[h]
	$i=1$\;
	\While{$i < n-k$}{
		\uIf{edge $(t_1^i, r_{k+1}^i)$ is in $E(T_{n-k+1})$}{
			remove edge $(t_1^i, r_{k+1}^i)$ from $E(T_{n-k+1})$\;
			add edges $(t_1^i, r_2^i), (r_2^i, t_2^i), (t_2^i, r_3^i), (r_3^i, t_3^i), \ldots, (r_k^i, t_k^i), (t_k^i, r_{k+1}^i)$\;
			$i=i+1$\;
		}	
		\Else{
				add the following edges to $E(T_{n-k+1})$:
				\begin{itemize}
				\item $(t_1^i, r_2^i), (t_1^i, r_3^i), \ldots, (t_1^i, r_{k}^i)$\; 
				\item $(r_j^i, t_j^i), (t_j^i, r_j^{i+1}), (r_j^{i+1}, t_j^{i+1}), \ldots, (r_j^{n-2}, t_j^{n-2}), (t_j^{n-2}, j)$ for $j=2, \ldots, k-1$\;
			
				\end{itemize}	
				\uIf{$k = n-1$}{
					add the following edges to $E(T_{n-k+1})$:
					\begin{itemize}
					\item $(r_k^i, t_k^i), (t_k^i, r_k^{i+1}), (r_k^{i+1}, t_k^{i+1}), \ldots, (t_k^{n-3}, r_k^{n-2}), (r_k^{n-2}, k)$\;
					\end{itemize}
				}
				\Else{
					add the following edges to $E(T_{n-k+1})$:
					\begin{itemize}
					\item $(r_k^i, t_k^i), (t_k^i, r_k^{i+1}), (r_k^{i+1}, t_k^{i+1}), \ldots, (r_k^{n-2}, t_k^{n-2}), (t_k^{n-2}, k)$\;
					\end{itemize}
				}
		}
	}
	\caption{} \label{alg}
\end{algorithm}

These operations transform the emebdding of $T_{n-k+1}$ into an embedding of $T_n$ in such a way that all vertices of $U_n$ are also vertices of the embedding of $T_{n}$ and the leaf sets of $U_n$ and $T_n$ coincide. Thus, $T_n$ is a base tree of $U_n$ (as an example see Figure \ref{Fig_Embedding}(i)). As $T_n$ was an arbitrary tree shape (that is not the star tree) on $n$ leaves and as the star tree is trivially a base tree for $U_n$ this completes the proof.

\end{proof}

\subsection{Unrooted universal treebased networks} \label{Sec_unrooted_universal}
Even though Theorem \ref{thm_rooted_universal} states the existence of a rooted non-binary universal treebased network on $n$ leaves for all positive integers $n$, we can use the same construction to show the following statement for unrooted networks.
\begin{corollary}
For all positive integers $n$, there exists an unrooted non-binary universal treebased network on $n$ leaves.
\end{corollary}
\begin{proof}
First, for  $n=1$ and $n=2$, an unrooted universal tree-based network trivially exists: For $n=1$, the only unrooted tree shape is a single vertex, which at the same time is the only tree-based unrooted network; for $n=2$, the only unrooted tree shape is an edge between the two leaves. Thus, any unrooted tree-based network on 2 leaves can be considered an unrooted universal tree-based network for $n=2$. Now, for $n \geq 3$, consider the construction of the non-binary universal tree-based network in the proof of Theorem~\ref{thm_rooted_universal}. By ignoring the designation of the vertex $\rho$ as root of the network and the orientation of edges, this construction yields an unrooted non-binary universal tree-based network with $n$ leaves, which completes the proof.
\end{proof}

\section{Discussion} \label{sec_discussion}  
In this note, we have constructively shown that there exist both a rooted and an unrooted non-binary universal tree-based network with $n$ leaves for all positive integers $n$. Like the rooted binary universal tree-based networks in \citet{Hayamizu2016, Zhang2016, Bordewich2017} the rooted non-binary universal tree-based network constructed in this paper is \emph{stack-free} (i.e. it has no two reticulations one of which is a parent of the other), but unlike the constructions in \citet{Hayamizu2016, Zhang2016, Bordewich2017} it is not \emph{temporal} (\emph{time-consistent}) (where a rooted network $N^r$ is called temporal if there is a mapping $t:V(N^r) \rightarrow \mathbb{N}$ such that if $(u,v)$ is a tree edge, then $t(u) < t(v)$, while if $(u,v)$ is a reticulation edge, then $t(u)=t(v)$). 
To see this, consider vertices $\rho, t_1^1$ and $r_2^1$ in $U_3$ (Figure~\ref{Fig_Embedding}(e)) or $U_5$ (Figure~\ref{Fig_Embedding}(i)). As $(\rho, t_1^1)$ is a tree edge, in a temporal network we would have $t(\rho) < t(t_1^1)$. However, as $(t_1^1, r_2^1)$ is a reticulation edge, we would also have $t(t_1^1) = t(r_2^1)$ and similarly, as $(\rho, r_2^1)$ is also a reticulation edge, we would have $t(\rho) = t(r_2^1)$. In particular by the previous case, $t(\rho) = t(r_2^1) = t(t_1^1)$, which contradicts the fact that $t(\rho) < t(t_1^1)$. 
It would thus be of interest for future research to investigate whether there also exists a rooted non-binary universal tree-based network on $n$ leaves that is temporal for all positive integers $n$. \\
Moreover, it might be of interest to infer a theoretical lower bound for the number of reticulations needed for a rooted non-binary tree-based network to be universal (as done in \citet{Bordewich2017} for the binary case) and to investigate if there exists a non-binary universal tree-based network \blue{that achieves this minimum}. Note that the construction given in this manuscript requires $(n-2)^2$ reticulations in the upper part $U_n$ and $n$ reticulations in the lower part (if the complete bipartite graph $K_{n,n}$ is used). Thus, the construction given here has $O(n^2)$ reticulations. However, our construction is possibly \blue{not minimal.}

\section*{Acknowledgement} 
We thank an anonymous referee of an earlier version of this manuscript for suggesting the simpler complete bipartite graph as an alternative to the Bene\v{s} network for the lower part of our construction. Moreover, we thank Lina Herbst for helpful discussions on the general topic and for comments on an earlier version of this manuscript.

\bibliographystyle{spbasic}\biboptions{authoryear}
\bibliography{References2}

\clearpage
\thispagestyle{empty}
\begin{figure*}[htbp]
\begin{subfigure}[b]{0.33\textwidth}
	\begin{center}
	\includegraphics[scale=0.2]{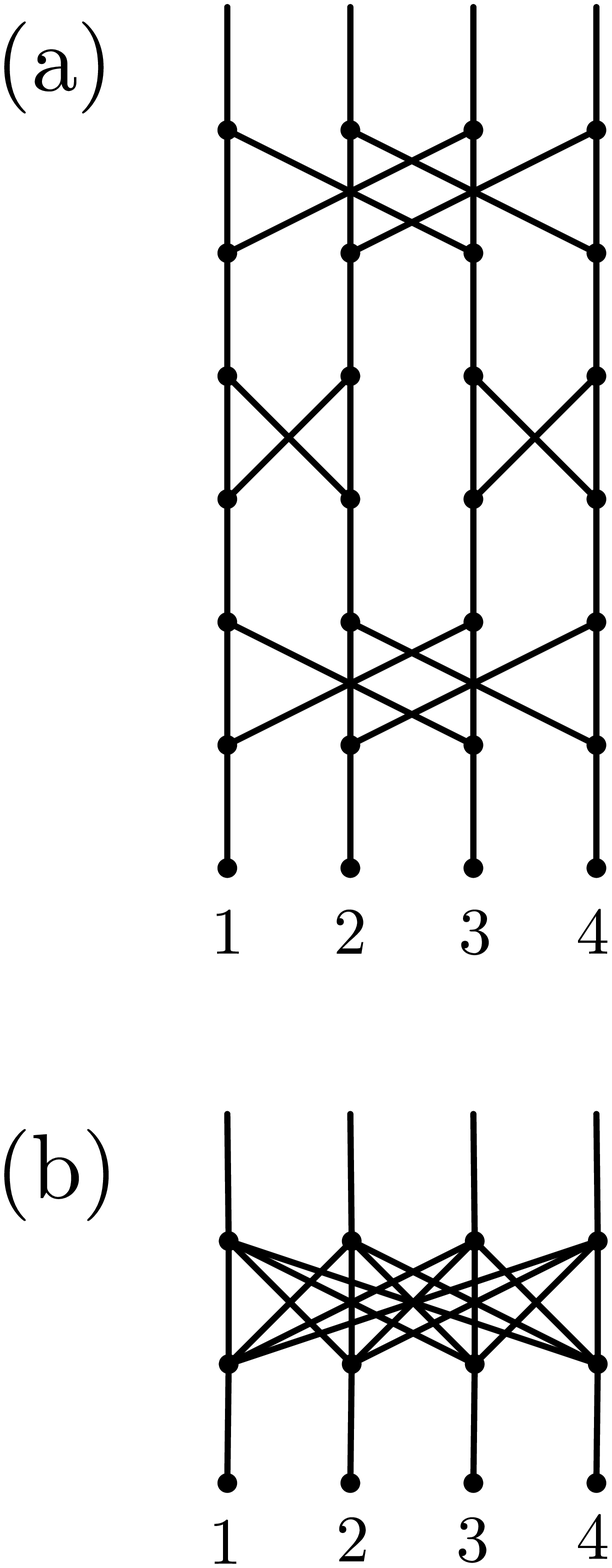}
	\end{center}
	\vspace*{5mm}
\end{subfigure}
\begin{subfigure}[b]{0.33\textwidth}
	\begin{center}
	\includegraphics[scale=0.2]{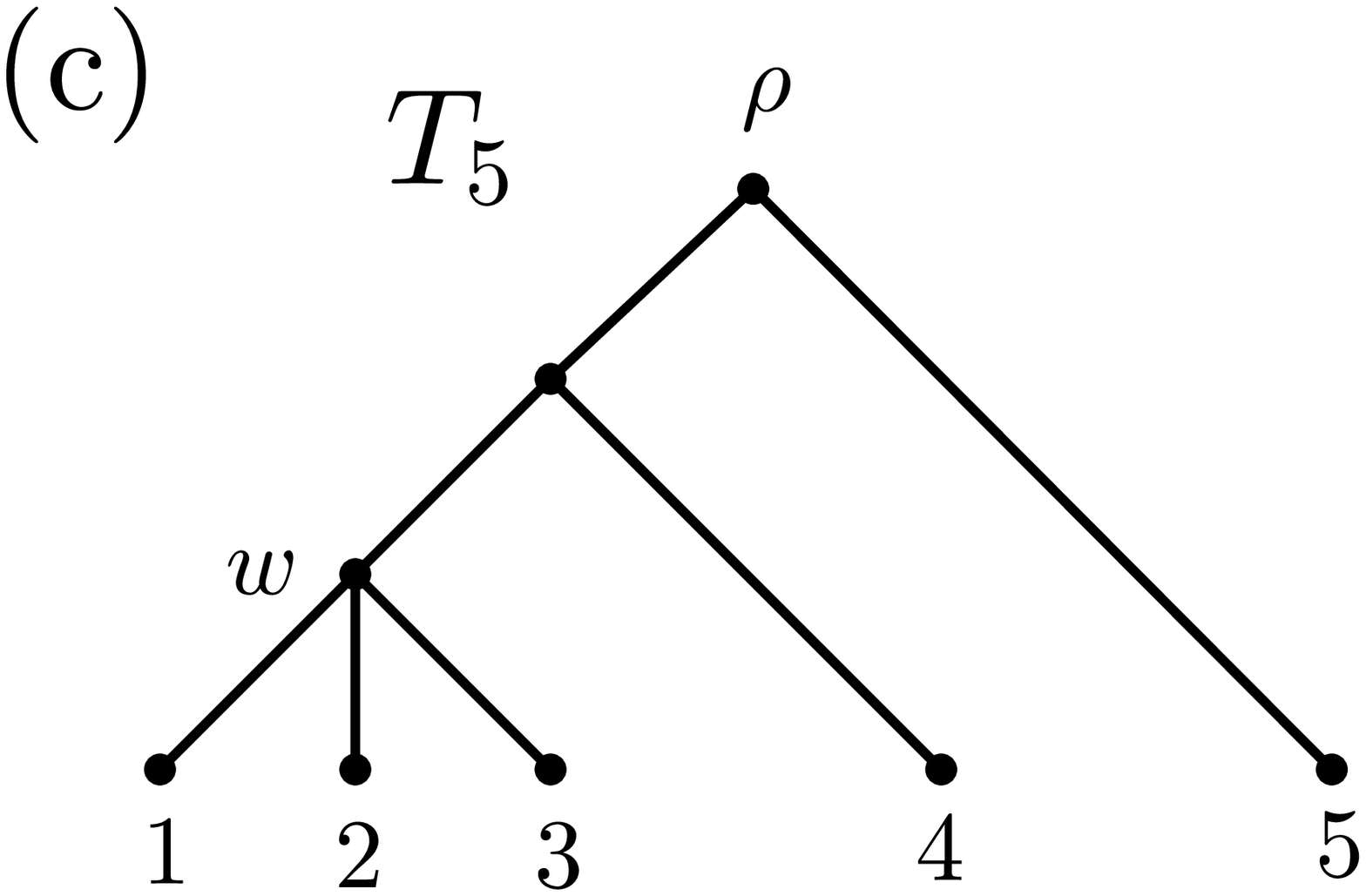} \vspace*{5mm} \\
	\includegraphics[scale=0.2]{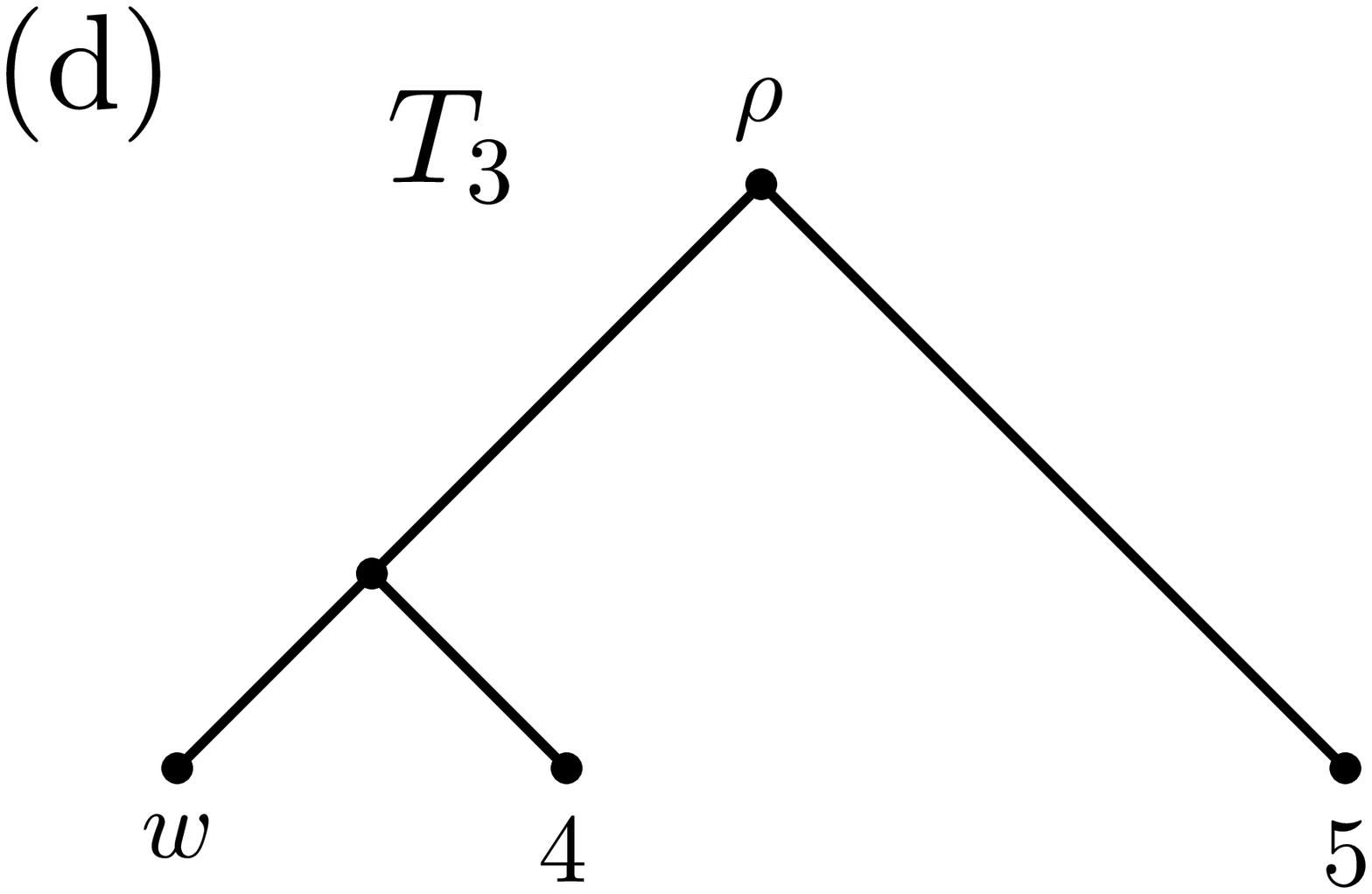} \vspace*{5mm} \\
	\includegraphics[scale=0.225]{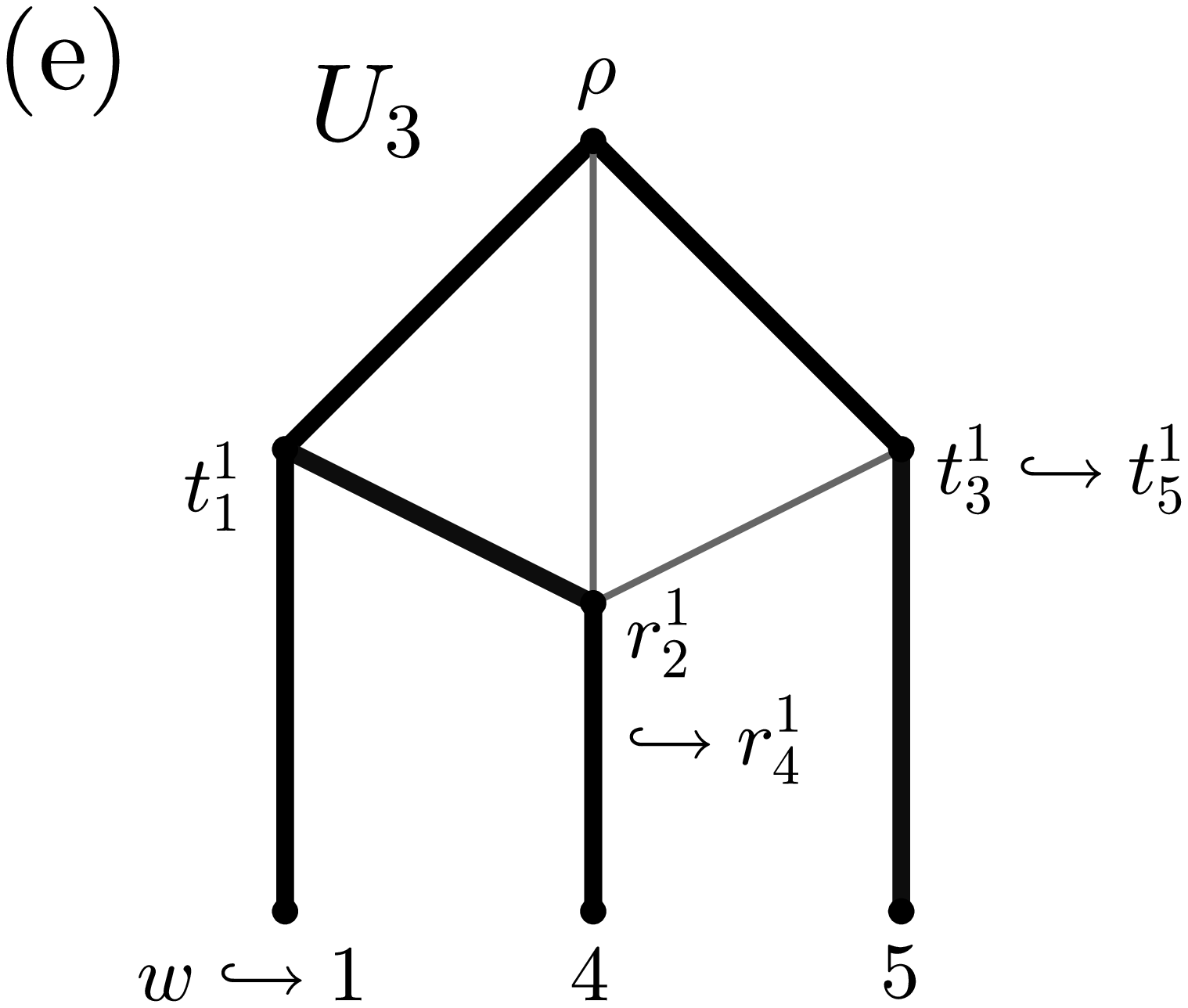}
	\end{center}
	\vspace*{5mm}
\end{subfigure}
\begin{subfigure}[b]{0.33\textwidth}
	\begin{center}
	\includegraphics[scale=0.215]{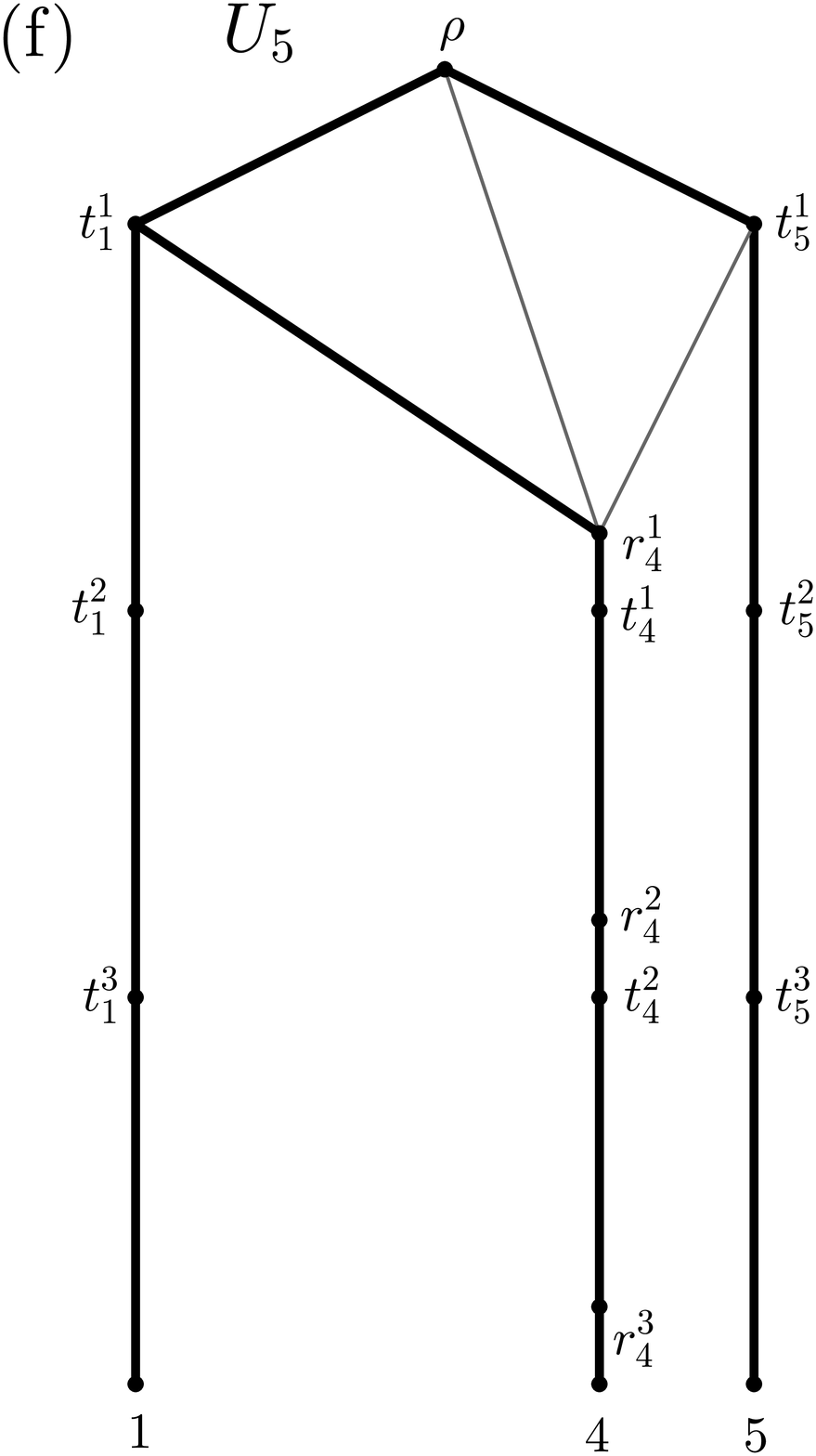} 
	\end{center}
	\vspace*{5mm}
\end{subfigure}
\begin{subfigure}[b]{0.33\textwidth}
	\begin{center}
	\includegraphics[scale=0.215]{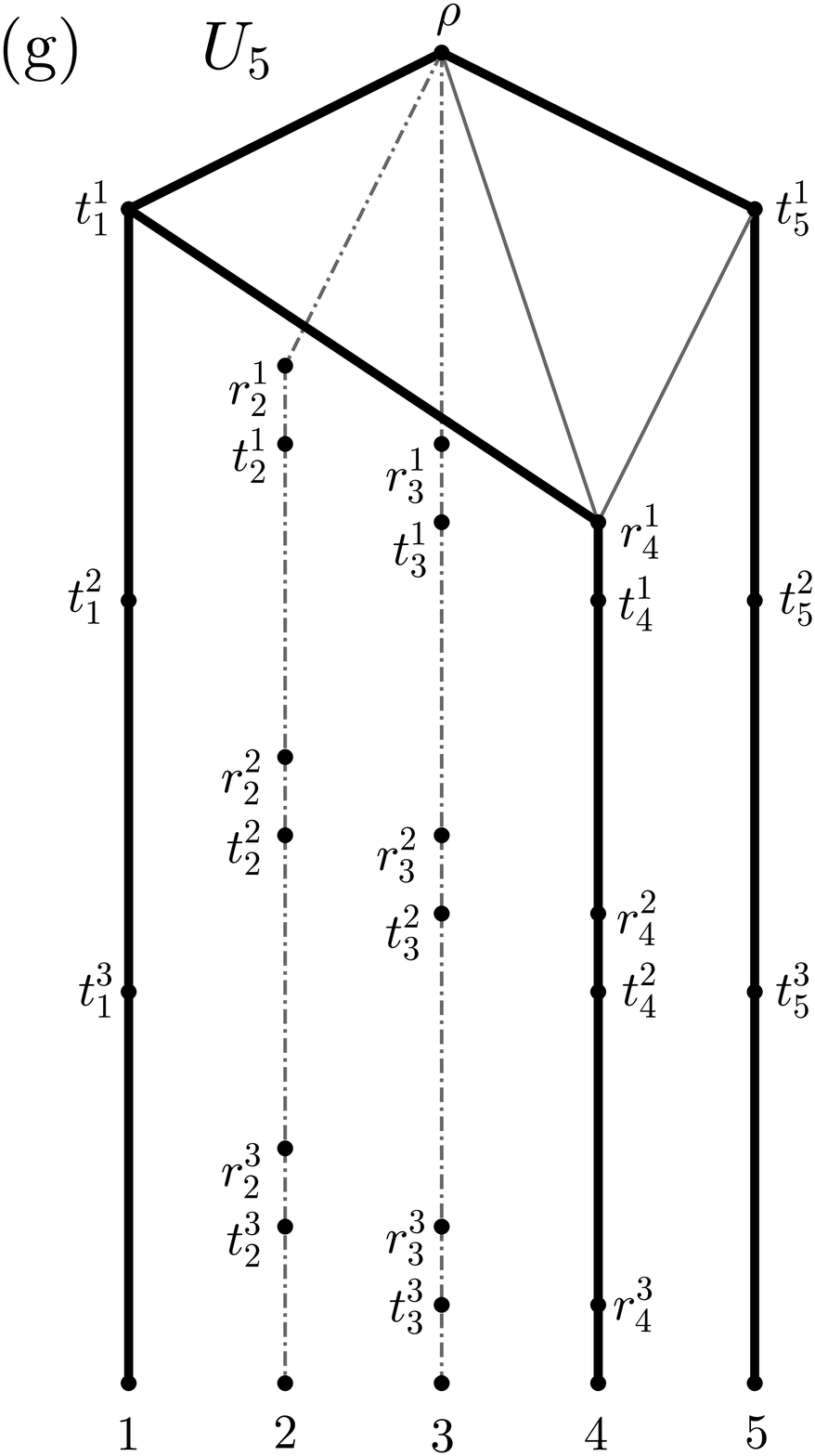} 
	\end{center}
	\vspace*{5mm}
\end{subfigure}
\begin{subfigure}[b]{0.33\textwidth}
	\begin{center}
	\includegraphics[scale=0.215]{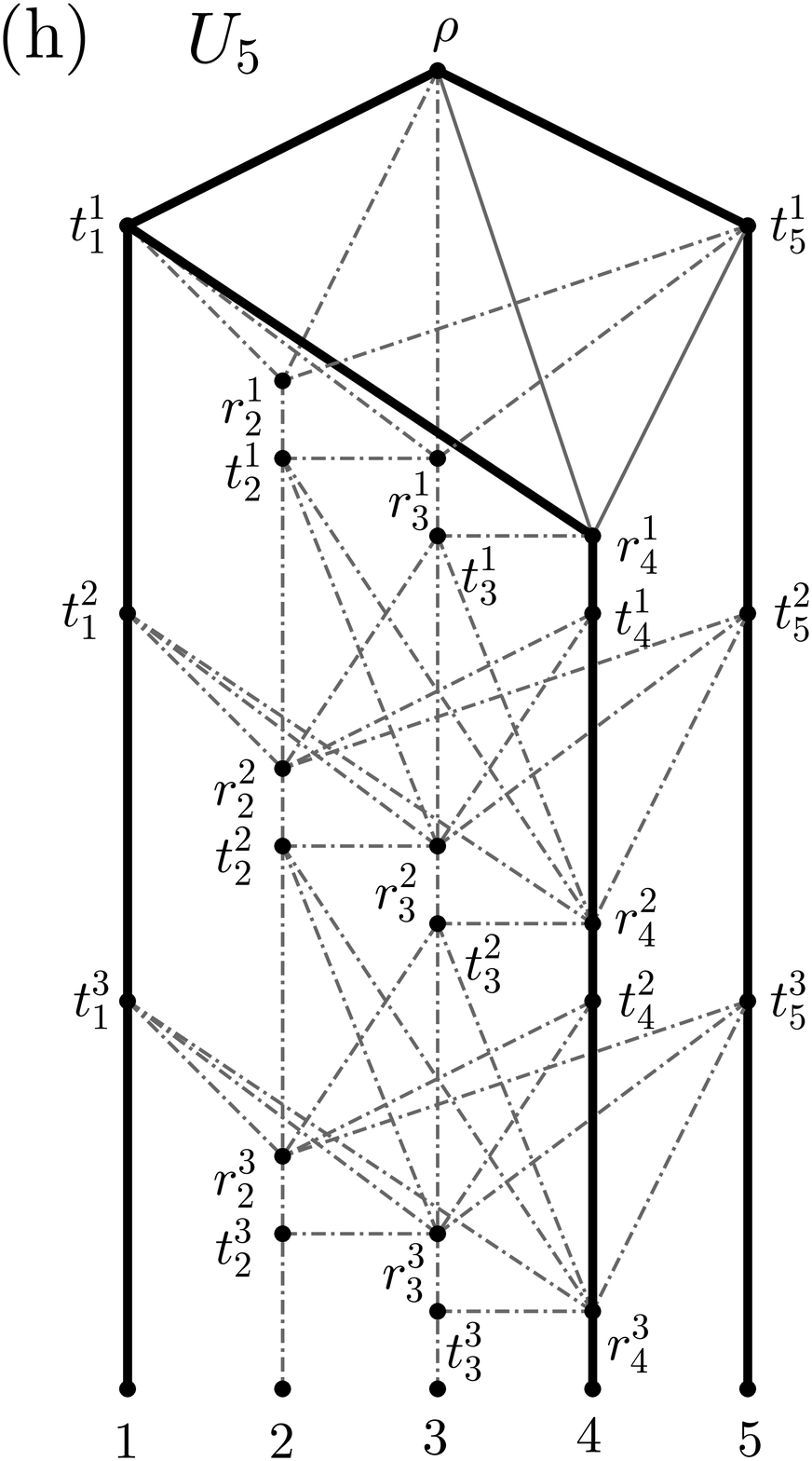} 
	\end{center}
	\vspace*{5mm}
\end{subfigure}
\begin{subfigure}[b]{0.33\textwidth}
	\begin{center}
	\includegraphics[scale=0.215]{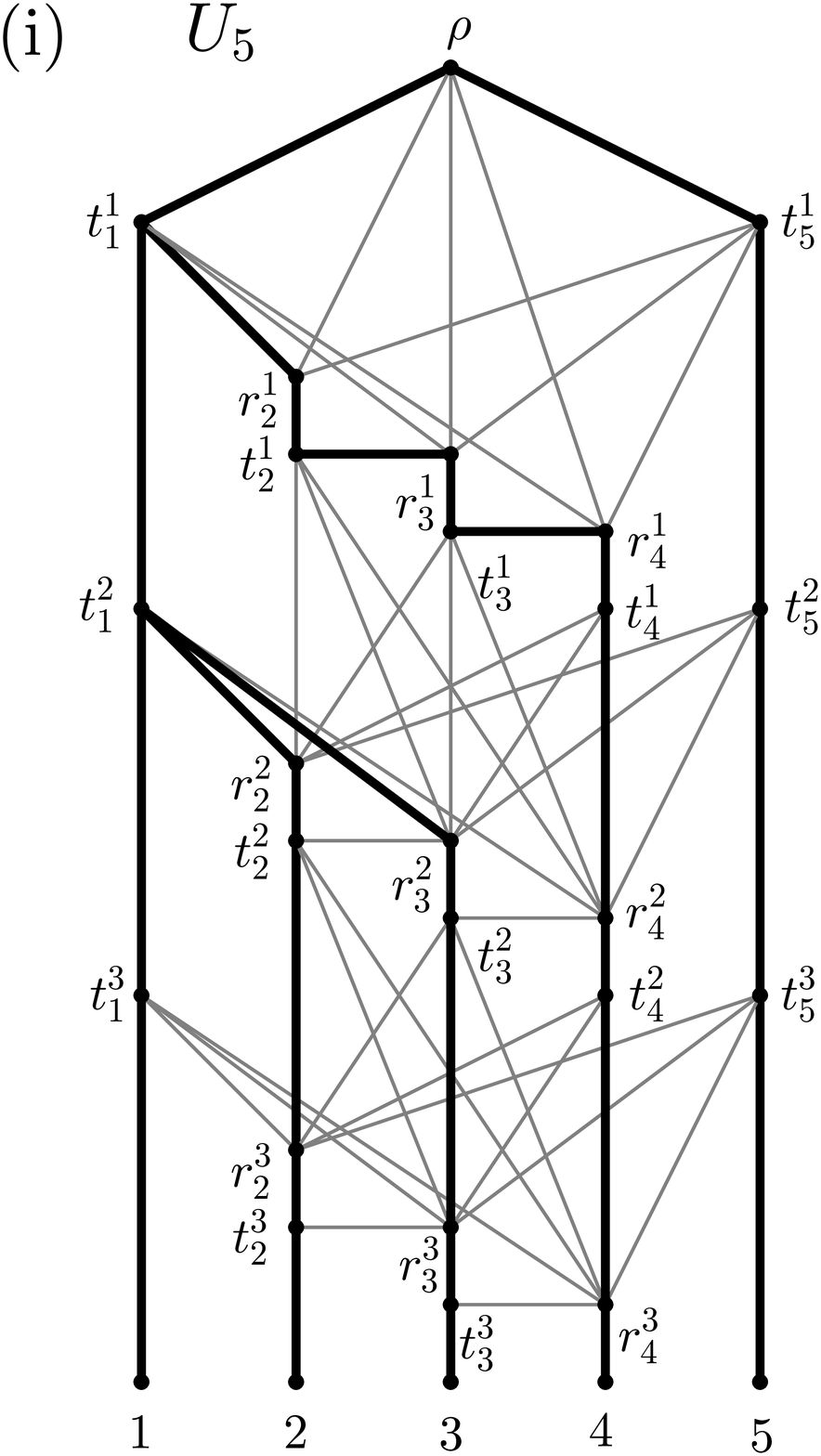} 
	\end{center}
	\vspace*{5mm}
\end{subfigure}
\caption{Illustration of the construction and concepts used in the proof of Theorem~\ref{thm_rooted_universal}. (a) and (b) show the lower part in the construction of a rooted non-binary universal tree-based network on 4 leaves in form of the Bene\v{s} network of size four (Figure taken from \citet{Bordewich2017}) (a) or the complete bipartite graph $K_{4,4}$ (b) (all edges are directed down the page). Moreover, $T_5$ is a non-binary tree shape on 5 leaves (c). We consider vertex $w$ and delete its children, which yields tree shape $T_3$ on 3 leaves (d). By the inductive hypothesis $T_3$ is a base tree of $U_3$; an embedding is depicted in bold (e). After relabeling vertices (e), 2 attachment points are added on the edges $(t_1^1,1)$ and $(t_5^1,5)$, respectively, and 4 attachment points are added on the edge $(r_4^1,4)$ (f). All new edges created in this step, e.g. $(t_1^1, t_1^2)$, are added to the embedding of $T_3$. Then, 2 edges connecting the root to leaves $2$ and $3$ are added. These edges are subdivided by introducing 6 attachment points on each edge (g). Then, the construction of $U_5$ is completed by introducing all missing edges between tree vertices and reticulation vertices and between pairs of reticulation vertices (h). In the last step, the embedding of $T_3$ is transformed back to an embedding of $T_5$ (i): Firstly, the edge $(t_1^1,r_4^1)$ (depicted in bold in (h)) is replaced by the edges $(t_1^1,r_2^1), (r_2^1, t_2^1), (t_2^1,r_3^1), (r_3^1, t_3^1)$ and $(t_3^1, r_4^1)$ (depicted in bold in (i)). In the last step the edges $(t_1^2, r_2^2), (t_1^2, r_3^2), (r_2^2, t_2^2), (t_2^2,r_2^3), (r_2^3, t_2^3), (t_2^3,2), (r_3^2, t_3^2), (t_3^2, r_3^3), (r_3^3, t_3^3)$ and $(t_3^3,3)$ are added to the embedding of $T_5$ into $U_5$. Note that in $U_5$ all horizontal edges, i.e. edges of type $(t_i^k, r_{i+1}^k)$ for $i=2,3$ and $k=1,2,3$, are directed left to right; all other edges are directed away from the root. Similarly, all edges in $T_3, T_5$ and $U_3$ are directed away from the root.}
\label{Fig_Embedding}
\end{figure*}

\end{document}